\documentclass[reqno]{article}
\usepackage{amsthm}
\usepackage{multicols}
\usepackage{multicol}
\usepackage{wrapfig}
\usepackage{amsfonts}
\usepackage{amsmath}
\usepackage{graphicx}
\usepackage{amssymb}
\usepackage{mathrsfs}
\usepackage{url}
\usepackage[margin=1in]{geometry}
\theoremstyle{definition}
\newtheorem{theorem}{Theorem}
\newtheorem{lemma}[theorem]{Lemma}

\newtheorem{proposition}[theorem]{Proposition}

\newtheorem{definition}{Definition}
\newtheorem{convention}[definition]{Convention}
\newtheorem*{openproblem}{Open Problem}

\newtheorem*{knightdarwin}{The Knight-Darwin Law}
\def\R{\mathbb{R}}

\begin{document}

\title{Infinite graphs in systematic biology, with an application to the species problem}
\author{Samuel A.~Alexander\thanks{Email:
alexander@math.ohio-state.edu}\\
\emph{Department of Mathematics, the Ohio State University}}
\date{June 2013}

\maketitle

\begin{abstract}
We argue that C.~Darwin and more recently W.~Hennig
worked at times under the simplifying assumption of an eternal biosphere.
So motivated, we explicitly consider the consequences which follow mathematically from this
assumption, and the infinite graphs it leads to.
This assumption admits certain clusters of organisms which have some
ideal theoretical properties of species, shining some light
onto the species problem.
We prove a dualization of a law of T.A.~Knight and C.~Darwin, and
sketch a decomposition result involving the internodons of D.~Kornet, J.~Metz and H.~Schellinx.
A further goal of this paper is to respond to B.~Sturmfels' question, ``Can biology lead to new theorems?''
\end{abstract}


\clearpage

\section{Introduction}

Dress et al.~(2010) recently renewed interest in the set of all organisms ever to have lived, 
endowed with a directed graph structure. We thought it a natural extension to consider the set 
$V$ of organisms (or other living things) which have ever lived or which will ever live in the 
future.
We found that this extension had already
been made explicit in Kornet (1993) and 
Kornet et al.~(1995), where the future is needed in order to distinguish between temporary 
and permanent splits in genealogical networks.

\subsection{Darwin as Infinite Graph Theorist}

With our future-oriented view of systematics in mind, we recalled that
particularly relevant passage from 
C.~Darwin's \emph{On the Origin of Species}:

\begin{knightdarwin} (Darwin 1872) ``...it is a general law of nature that no organic being 
self-fertilises itself for a perpetuity of generations; but that a cross with another individual 
is occasionally-- perhaps at long intervals of time-- indispensable.''
\end{knightdarwin}

In (Francis Darwin 1898) we learn that the above principle is called the Knight-Darwin Law.
But in a world where life goes extinct in finite time, 
this law is trivial:
even if all reproduction in the entire world were asexual, the Knight-Darwin Law would 
still trivially hold for lack of an instance of ``perpetual'' (infinite) self-fertilisation to 
falsify it.

\begin{wrapfigure}{r}[.5in]{0.5\textwidth}
\vspace{-10pt}
\begin{center}
\includegraphics[scale=.7]{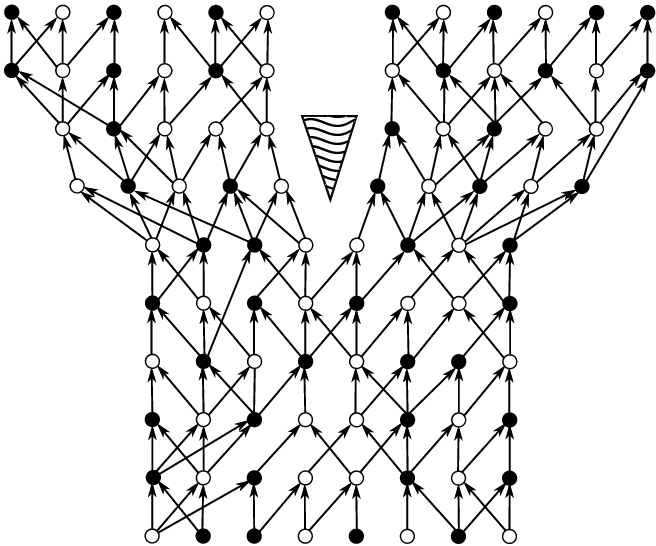}
\end{center}
\vspace{-10pt}
\begin{quote}
\caption{A reproduction of a graph in W.~Hennig's (1966) book (Hennig's 
Figure 4, p.~19).
If we understand correctly, the triangle denotes the permanency of the
corresponding cleavage; this triangle would be
superfluous if Hennig did not intend
us to imagine the graph continuing beyond what is shown.}
\end{quote}
\vspace{-.9in}
\end{wrapfigure}
Thus not only is Darwin concerned with the distant future of life, he seems to admit 
the 
possibility of an infinite biosphere-- the possibility that life will never go extinct-- 
because if he explicitly thought this impossible, then
there would be no point suggesting the above Law.

Based on a careful reading of surrounding passages, we believe the Knight-Darwin Law can be 
glossed in modern graph-theoretical language as follows (hereafter, $G$ is the graph of all 
living organisms, past and future, an arc directed from $u$ to $v$ precisely if $u$ is a parent 
of $v$):
\begin{knightdarwin} \emph{(Graph-theoretical version)} The graph $G$ does not contain an
infinite directed path of vertices each of which has $<2$ parents.
\end{knightdarwin}

Darwin did know of apparent counterexamples (see Francis Darwin 
1898), reducing the Knight-Darwin Law to a simplifying 
assumption; nevertheless, a simplifying assumption of remarkable
graph-theoretical sophistication for its time.
One might argue 
that infinite graphs were not well studied outside of cutting-edge research mathematics
until (K\H{o}nig 1936), many decades after 
Darwin's work.

In a later section we will prove that if true, the Knight-Darwin Law 
logically implies a Dual Knight-Darwin Law.

\subsection{Hennig as Infinite Graph Theorist}

There is compelling reason to suspect the Figure 4 on p.~19 of W.~Hennig's 
(1966) book
(see our Figure 1) implicitly depicts an infinite graph--
joining Hennig with Darwin in infinite graph theory.

The figure in question illustrates a \emph{cleavage} in a biological network.
Hennig refers (in his caption) to ``the \emph{process} of species cleavage'' (emph.~mine),
suggesting a dynamic continuation of the genealogical network (graph) splitting up.
Thus (we believe) Hennig must intend the population to continue beyond
what is shown (what is shown is too ephemeral to 
be a \emph{process}).
A finite continuation would have the same problem, thus
suggesting an infinite intended extension.
In the text accompanying the figure, Hennig speaks of
species which ``persist over long periods of time'' but which are ``not absolutely permanent''.
Like the Knight-Darwin Law, there would be no reason to even mention absolute permanence
if Hennig assumed a finite end to the biosphere.

Furthermore, Hennig seems to distinguish between one particular cleavage
(indicated by a triangle) and smaller cleavages also present in the graph.
Surely the distinction cannot be as arbitrary\footnote{Baum \& Shaw (1995, p.\ 294)
and Velasco (2008, p.\ 868) point out the problem of \emph{big} cleavages and \emph{small} cleavages, and the
resulting vagueness.  This vagueness disappears under infinitary assumptions.}
as ``branch size 26'' versus ``branch size 1''.
Hennig discounted the minor cleavages because he meant us to imagine them as temporary--
they would disappear if we just saw a little more of the graph.
But any finite number of additional generations would still suffer minor cleavages\footnote{Unless
something unnatural occurred, for example, all the living specimens of each species joining to co-parent a
single sterile child.}
and the only way to eliminate them is to continue the graph infinitely far.



\subsection{An Application to the Species Problem}

The \emph{species problem} is the problem of
finding a proper definition for the intuitive concept of a 
biological species.
Many species notions exist, each with its own pros and cons, and none have
managed to reach universal acceptance.

By studying $G$ with an explicit infinitude assumption, we have arrived at some cluster
notions (described in Sections 3 and 4) which we have
decided to call \emph{infinitary genera} and \emph{infinitary species}.
These name choices will be justified below; \emph{infinitary genus} is more or less an arbitrary name
and could be replaced by \emph{infinitary family} or \emph{infinitary order} or
\emph{infinitary tree-of-life-node}; however, the name \emph{infinitary species} is important.

We won't pretend that our species notion is like any species notion used by everyday
biologists; it would be useless in the field, and we would never dream of suggesting
it as a replacement for the practical species notions.
However, we hope our notion will apply to the species problem
in three ways: by offering a solution to 
a
species sorites paradox;
by theoretically reconciling competing morphological and
non-morphological approaches to
species definition;
and by telling us something about the structure of the far future 
extremities of the tree of life.

\section{Some Further Justification for Infinitary Assumptions}

\begin{quote}
``So profound is our ignorance, and so high our presumption,
that we marvel when we hear of the extinction of
an organic being; and as we do not see the cause,
we invoke
cataclysms
to desolate the world, or invent laws on the duration of the forms of life!''  --Charles Darwin (1872)
\end{quote}

In assuming that infinitely many individuals will live, we are hardly the first scientists to
approximate the finite by the infinite.  The assumption is very similar
to how physicists and chemists
study symmetry groups of crystal patterns or tilings.  We assume such patterns fill space to infinity,
because otherwise it would be very unnatural, if possible at all, to sensibly talk about their translation
symmetries--  we wonder whether this might somewhat explain, by analogy, why species are so difficult
to define.

\subsection{Justification by Pursuit of New Mathematics}

By assuming infinity we will obtain
somewhat unique responses to Sturmfels' (2005) question, ``Can biology lead to new
theorems?''  If biology is to be ``mathematics' next physics, only better'' (Cohen 2004)
it must generously contribute to the combinatorial and infininitary branches of mathematics,
as physics has done.

One way to distinguish mathematical theorems is on the strength of the set-theoretical
assumptions which they require.  To be sure, biology has already contributed much to
mathematics, but the author is unaware of any theorems from biology which hinge on the \emph{Axiom
of Choice}; we will exhibit such a theorem in Section 4 
(Theorem~\ref{maintheorem}).



\subsection{Reduced Dependence on Scale}

There seems to be some disagreement on whether the systematist ought
to focus on individual living organisms, or whether to zoom out and 
consider larger populations as atoms; see, for example, the back-and-forth 
between de Queiroz \& Donoghue (1988) (pro-individual organisms) and Nixon 
\& Wheeler (1990) (the opposite).

If the biosphere is finite, this scaling decision has a big effect on 
the 
shape of the biosphere.
If one systematist takes organisms as 
atoms, and another takes some populations approximating 
species, the sizes of the resulting biospheres differ quite a bit.

On the other hand, if both systematists
operate under the assumption of an infinite tree of life,
their decisions do not effect the shape of the biopshere: viewed through 
either lens, the biosphere is infinite.

Note that this scaling decision can go both ways.  Rather than considering 
individual organisms as the vertices of our graph, we could zoom 
\emph{inward} and consider (say) individual X-chromosomes as vertices.  In a 
finite world, this decision would massively blow up our graph $G$, but 
under infinitary assumptions, it makes no major difference.

\subsection{A word to the most hardcore finitist}

I hope my paper will be useful to you, too.
In later sections there are results (e.g.~Proposition 6) which are,
I think, counter-intuitive enough that you might
be able to get away with calling them paradoxes, and using them to advance the finitist argument
by \emph{reductio ad absurdum},
in the same way a critic of the Axiom of Choice might use the Banach-Tarski paradox.
In other words, feel free to read the paper as a set of theorems specifically intended
to \emph{disprove} infinitary assumptions.

\section{Infinitary Genera}

In this and the next two sections, we will exhibit some interesting 
cluster notions which arise from the infinite biosphere assumption
along with some additional assumptions which we consider less scandalous.
The precise assumptions we make are as follows.

\begin{itemize}
\item (A1) We assume (like Dress et al.~(2010)) there are only finitely 
many 
\emph{roots}, that is,
parentless individuals.
\item (A2) We assume no individual is a parent of infinitely many 
children.
\item (A3) We assume each vertex $v\in V$ has a \emph{birthdate} 
$t(v)\in\R$, that $t(u)<t(v)$ whenever $u$ is a parent of $v$,
and that for every real $x$, $\{v\in V\,:\,t(v)<x\}$ (the set of 
individuals born before time $x$) is finite.
\item (A4) We assume $G$ is infinite, that is, infinitely many individuals 
will live.
\end{itemize}

We will define an infinitary genus to be an infinite set of individuals 
which is \emph{closed under ancestry}, that is, which contains every 
ancestor of every member of itself.  But it would be rash to make this 
definition without first motivating it.
In mathematics, it is a sign of a notion's importance when it arises 
unexpectedly from seemingly-unrelated
competing notions.  We will arrive at our infinitary genus notion in 
precisely this way.

To be clear, we use the word \emph{genus} arbitrarily where any of 
\emph{family}, \emph{order}, or \emph{tree-of-life-node}
would work just as well.
Mathematical language
is great for speaking about things which are maximal (e.g.~the entire 
biosphere) and things which are minimal
(e.g. infinitary species, as we'll see in the next section) but not at 
distinguishing between different intermediate
levels (e.g.~genera vs.~families).

With the above paragraphs in mind, let us make the following attempt at 
defining a group of individuals.
We will define what we call a \emph{birthdate genus}.
Our hope is that this is particularly obvious, something the reader
could easily invent on their own.
We will then prove it to be equivalent to our desired infinitary genus 
notion, in sight of Assumptions A1--A4.
Thus the following definition will not itself feature prominently in the 
rest of the paper, its purpose
is to motivate a definition to come after.

\begin{definition}
\label{basicspecies}
By a \emph{birthdate genus}
I mean a set $S$ of individuals such that there is some time $t$ such that 
every member of $S$ has
birthdate $\geq t$ and every external ancestor (that is, every individual 
outside $S$ but with a descendant
in $S$) has birthdate $<t$.
\end{definition}

As a clustering notion, this
is a weakened version of the much stronger notion which Dress et al.~call 
the \emph{Apresjan cluster},
following Steel (2007).
It captures the intuition that a node in the Tree of Life
must have evolved into existence at some particular time, and thus,
all its members are at least that young, while every external ancestor is 
strictly older.

\begin{convention}
\label{conv1}
If $X$ and $Y$ are infinite sets, we say $X$ is \emph{almost equal to} $Y$ 
(and write $X\approx Y$)
if their symmetric difference $(X-Y)\cup (Y-X)$ is finite.
\end{convention}

The whole reason to define birthdate genera was to allow the
following motivating result:

\begin{proposition}
\label{motivate}
If $S$ is an infinite set of individuals, then the following two 
conditions are equivalent:
\begin{enumerate}
\item $S$ is almost a birthdate genus.
\item $S$ is almost its own ancestral closure.
\end{enumerate}
\end{proposition}

\begin{proof}
Let $\bar{S}$ be the ancestral closure of $S$.

$(1\Rightarrow 2)$ Assume $S$ is almost a birthdate genus.
Thus, there is a birthdate genus $S'$ such that $S\approx S'$.
There is some time $t$ such that
all members of $S'$ are born at time $\geq t$ and all external ancestors 
of $S'$ are born at time $<t$.

Now, we claim that $\bar{S}-S$ is finite.
Suppose $u\in \bar{S}-S$.
Since $u\in\bar{S}$, $u$ is an ancestor of an individual $v\in S$.
There are three cases:
\begin{itemize}
\item Case 1: $u\in S'-S$.  Since $S'-S$ is finite, this can only occur 
for finitely many $u$.
\item Case 2: $u\not\in S'-S$ and $v\in S'$.
Then $u$ is an external ancestor of $S'$, so is born before time $t$.  By 
Assumption A3, this can
only occur for finitely many $u$.
\item Case 3: $u\not\in S'-S$ and $v\not\in S'$.
Since $v\in S$, $v\in S-S'$.
Since $S-S'$ is finite, there are only finitely many possibilities for 
$v$, and Assumption A3
implies that those finitely many possibilities only have finitely many 
ancestors, so again,
Case 3 can only occur for finitely many $u$.
\end{itemize}
This shows $\bar{S}-S$ is finite.  Since $S-\bar{S}=\emptyset$ is also 
finite, $S\approx \bar{S}$ as desired.

$(2\Rightarrow 1)$
Conversely, assume $S\approx \bar{S}$.
Thus, there are only finitely many external ancestors of $S$, call them 
$S_1$.
Of these finitely many individuals, let $t$ be the latest birthdate which 
occurs.
By Assumption A3, only finitely many individuals were ever born as of time 
$t$, call them $S_0$.
We claim $S-S_0$ is a birthdate genus.
Certainly any member of $S-S_0$ is born after time $t$, by definition of 
$S_0$.
Now suppose $v\not\in S-S_0$ has a descendant in $S-S_0$.  It could be 
$v\in S_0$,
in which case $v$'s birth date is no later than $t$, by definition of 
$S_0$.
Otherwise, $v$ must be outside of $S$, and thus, having a descendant in 
$S$,
$v$ is in $S_1$, so that (by choice of $t$), again $v$ is born no later 
than $t$.
\end{proof}

If we ignore finite sets as being insignificant, and treat infinite 
birthdate genera as being identical to sets which are \emph{almost}
infinite birthdate genera, we arrive (via Proposition~\ref{motivate}) at a 
much simpler definition,
and one which is easier to work with,
at the price of an error which is
infinitesimal.

\begin{definition}
\label{def1}
An \emph{infinitary genus} is an infinite, ancestrally closed set of 
individuals.
\end{definition}

Right away we would like to demonstrate how closely interwoven 
Definition~\ref{def1} is with
the point-set topology of cladistics: further evidence of the legitimacy 
of the
definition and its worthiness of study.

\begin{proposition}
\label{propclade}
An infinite set $S$ is an infinitary genus if and only if it is a closed 
set in the topology on $V$ where
clades are the basic open sets (by a \emph{clade} of course we mean an 
individual and
its descendants).
\end{proposition}

\begin{proof}
In other words, we must show that for infinite $S$, $S$ is an infinitary 
genus if and only if its complement $S^c$ is a union
of clades.

$(\Rightarrow)$
Assume $S$ is an infinitary genus.
Let $v\in S^c$.  Since $S$ is ancestrally closed, $v$ cannot be an 
ancestor of any element of $S$.
Thus, the clade $C_v$, consisting of $v$ and all its descendants, is 
disjoint from $S$.
Therefore $S^c$ is the union
\[S^c=\bigcup_{v\in S^c} C_v\]
of clades.

$(\Leftarrow)$
Assume $S^c$ is a union of clades and let $v\in S$ have an ancestor $u$; 
we must show $u\in S$.
If $u$ were not in $S$, $u$ would be in $S^c$, hence in some clade 
entirely contained in $S^c$.
But $v$ would be in that clade, forcing $v\in S^c$, a contradiction.
\end{proof}


\section{Infinitary Species}

The motivation for our species notion is the conviction that species 
should be the smallest taxa--
below genera, families, orders and so on.
Thus, we will define an infinitary species to be an infinitary genus in 
which no proper subset is an infinitary genus.

But why does such an obvious-seeming approach work for us when it has not 
worked for the finitist?
If a set of organisms constitutes a species, then presumably it would 
still constitute a species if one specimen were
eliminated.  According to our minimalist approach, if we could throw that 
specimen out of the species and still have
a species, then we must do so.  In a finite world, this repeated act of
discarding specimens would render the species empty.  This is an instance 
of the sorites paradox, or paradox of the heap
(Hyde 2011).
However, there is a way out of the trap:  if we require our species to be 
infinite and ancestrally closed,
the premise that we can always discard one specimen and retain a species 
gains a caveat:
we cannot discard the specimen if it has so many descendants that doing so 
would render the species finite.
We shall see that this causes the sorites paradox to vanish.

\begin{wrapfigure}{r}[.1in]{0.5\textwidth}
\vspace{-15pt}
\begin{center}
\includegraphics[scale=.9]{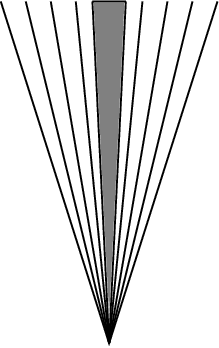}
\end{center}
\vspace{-10pt}
\begin{quote}
Figure 2: A species (shaded) nested within a larger genus, within a larger 
order, within a larger family, and so on.
\end{quote}
\end{wrapfigure}

\begin{definition}
An \emph{infinitary species} (or more simply an \emph{inspecies}) is an 
infinitary genus in which no
strictly smaller infinitary genus is a subset.
\end{definition}

A priori, there might be no inspecies-- either because of the sorites 
paradox, or because
infinitary genera can be refined
further and further with no end.
The main technical theorem of this paper is:

\begin{theorem}
\label{maintheorem}
Given assumptions A1--A4, there is at least one inspecies.
\end{theorem}

The proof will use the Axiom of Choice, in its Zorn's Lemma (Zorn 1935, 
Gowers 2008) form\footnote{By A3,
the biosphere is countable; it can be shown that only the \emph{countable} 
axiom of choice is needed
to guarantee existence of inspecies.}.
We will briefly review Zorn's Lemma before proving the theorem
(we state the lemma in a form most suitable for the use to which we will 
apply it).
A reader willing to take our word for it can safely skip the proof,
but bear in mind that if any one of A1, A2, A3, or A4 were denied, the 
other three
would not imply the theorem.

Recall that a binary relation $\supseteq$ on a set $\mathscr{X}$ is a 
\emph{partial order}
if it is reflexive, transitive, and anti-symmetric (anti-symmetry means 
that if $X\supseteq Y$ and $Y\supseteq X$
then $X=Y$).
A family $(Z_i)_{i\in I}$ of elements of $\mathscr{X}$, indexed by a 
linear order $I$, is a \emph{$\supseteq$-chain} if $Z_i\supseteq Z_j$
whenever $i\leq j$.
A \emph{bound} for this chain is an element $Z\in\mathscr{X}$ such that 
$Z_i\supseteq Z$ for every $i\in I$.
An element $X\in\mathscr{X}$ is \emph{extremal} if the only 
$Y\in\mathscr{X}$ such that $X\supseteq Y$ is $Y=X$ itself.

\begin{theorem}
\emph{(Zorn's Lemma)}
Let $\supseteq$ be a partial order on a nonempty set $\mathscr{X}$.
Suppose that every nonempty $\supseteq$-chain has a bound.
Then $\mathscr{X}$ has an extremal element.
\end{theorem}

Armed with this mathematical logical sledgehammer, we can prove 
Theorem~\ref{maintheorem}.

\begin{proof}[Proof of Theorem~\ref{maintheorem}]
We may assume the following $(*)$:
that every non-root individual is a descendant
of \emph{every} root.  This is safe to assume because if it is untrue, we 
can (thanks to A1) add an imaginary ``super root'' to
our graph and declare it to be the lone parent of all actual roots.

Let $\mathscr{X}$ be the set of all infinitary genera.  Notice 
$\mathscr{X}\not=\emptyset$
since it contains $V$ itself: the set of \emph{all} individuals is an
infinitary genus by Assumption A4.
Notice that the superset relation $\supseteq$ partially orders 
$\mathscr{X}$, and an infinitary genus $X\in\mathscr{X}$
is an inspecies precisely if it is extremal with respect to $\supseteq$.
Therefore, by Zorn's Lemma, it is sufficient to let $(Z_i)_{i\in I}$ be an 
arbitrary nonempty $\supseteq$-chain from $\mathscr{X}$
and show it has a bound in $\mathscr{X}$.  We will show that $Z=\cap_{i\in 
I} Z_i$ is such a bound.
Obviously for every $j\in I$, $Z_j\supseteq \cap_{i\in I} Z_i$, so all 
that remains is to show
$\cap_{i\in I} Z_i$ is in $\mathscr{X}$ (i.e., that it is an infinitary 
genus).

First the easy part: we show that $\cap_{i\in I} Z_i$ is ancestrally 
closed.
Let $v\in \cap_{i\in I} Z_i$ and let $u$ be an ancestor of $v$.
For every $i\in I$, we have $v\in Z_i$, and $Z_i$ is ancestrally closed, 
so $u\in Z_i$;
by arbitrariness of $i$, this shows $u\in \cap_{i\in I} Z_i$, establishing 
ancestral closure of $Z$.

The difficult part is to show that $Z$ is infinite.
Assume, for sake of contradiction, that $Z$ is finite.

Let $q_1,\ldots,q_m$ be those individuals who have a parent 
in $Z$ but
who are not in $Z$ themselves: there are finitely many of these by
Assumption A2.
For any $1\leq k\leq m$, the fact that $q_k\not\in Z=\cap_{i\in I}Z_i$ 
means there is
some $i_k\in I$ such that $q_i\not\in Z_{i_k}$, and thus (since 
$(Z_i)_{i\in I}$ is a chain)
more strongly $q_k\not\in Z_{i}$ whenever $i\geq i_k$.
Therefore, letting $i=\max\{i_1,\ldots,i_m\}$,
we have: $q_1\not\in Z_i$, $q_2\not\in Z_i$, $\ldots$, $q_m\not\in Z_i$.

Now, $Z_i$ is infinite since it's an infinitary genus.  Thus $Z_i-Z$ is 
infinite since $Z$ is finite.
In particular, $Z_i-Z$ is nonempty.  Thus, we can pick an individual $r\in 
Z_i-Z$ with shortest possible
reverse-path to a root
(every individual has a reverse-path to a root by Assumption A3).  By 
$(*)$, $Z$ contains all roots, so $r$ is not 
itself a root, and that shortest path
is not the empty path.  So $r$ has a parent on that minimal-length path.  
This parent \emph{must} be in $Z$,
because otherwise, we would have chosen the parent instead of $r$ (the 
parent has a shorter path to a root).

I've shown that $r$ has a parent in $Z$, but $r$ itself was chosen outside 
$Z$.  By definition this means $r$
is one of the $q_1,\ldots,q_m$.  This is nonsense, because $r\in Z_i$ and 
all of the $q_1,\ldots,q_m$
are absent from $Z_i$.

By contradiction, $Z$ is infinite, and so (being ancestrally closed) it is 
an infinitary genus.
I've shown that an arbitrary nonempty chain has a bound; by Zorn's Lemma, 
$\mathscr{X}$ has an extremal element,
i.e., there is an inspecies.
\end{proof}

\section{Results}
\label{section4}

The key to the structural properties of inspecies is the following result,
remarkable for being so powerful while having such a simple proof.

\begin{proposition}
\label{structureprop}
Suppose $C$ is an inspecies and $S\subseteq C$ is any infinite subset.
Then every individual in $C$ has a descendant in $S$.
\end{proposition}

\begin{proof}
Let $C'$ be the set of ancestors of individuals in $S$.
Then $C'$ is clearly ancestrally closed;
since $S$ is infinite, Assumption A2 implies $C'$ is infinite;
altogether, $C'$ is an infinitary genus.
And $C'\subseteq C$ since $C$ is ancestrally closed,
so $C'=C$ by minimality.
\end{proof}

This has remarkable consequences for individuals within inspecies.  Take 
any property $P$,
which may hold of some individuals and not hold of other individuals.  
Within an inspecies,
every individual has a descendant with property $P$, or every individual 
has
a descendant without property $P$ (or both).  For instance, in an 
inspecies everybody has a vertebrate descendant,
or everybody has an invertebrate descendant.  If there is a fixed 
universal upper bound on the number of hairs an
individual can have on their body, then for any inspecies, there is some 
number $n$ such that everybody in the inspecies
has a descendant with exactly $n$ hairs on their body.
If we let Bertrand Russell choose the property $P$, he would no doubt 
choose
the property ``not a descendant of Bertrand Russell'' and thereby lead us 
to:

\begin{proposition}
\label{prolific}
In an inspecies, every individual is an ancestor of almost every 
individual.
\end{proposition}

\begin{proof}
Let $u$ be a member of an inspecies and let $S$ be the set of members of 
that inspecies who are not descended from $u$.
If $S$ were infinite, then by Proposition~\ref{structureprop}, $u$ would 
have a descendant in $S$, which is absurd.
So $S$ is finite, meaning almost every member of the inspecies descends 
from $u$.
\end{proof}

With Proposition~\ref{prolific} in our arsenal we can productively compare 
the inspecies
with the \emph{tight cluster} notion of Dress et al.
If $C$ is an inspecies and we let
$D(\supset_{\approx} C)$ consist of all individuals in $V$ whose
descendants include \emph{almost} all of $C$, then 
Proposition~\ref{prolific} implies
$D(\supset_{\approx} C)=C$.
Therefore, every inspecies is a kind of one-sided version
of a tight cluster (it is one-sided because the clusters of Dress et 
al.~(2010) are designed to be closed
descendantially (at least when unborn future individuals are ignored) and 
ours are not).
A similar observation goes for
Baum's (2009) concept of \emph{organismic
exclusivity} (as described by Dress et al.)

The next proposition will theoretically reconcile morphological and 
non-morphological approaches to the species problem.

\begin{proposition}
\label{trichotomy}
\emph{(Morphological Trichotomy)}
Let $C$ be an inspecies and let $P$ be a property of individuals.  Exactly 
one of the following is true:
\begin{enumerate}
\item $P$ holds of almost every member of $C$,
\item $P$ fails of almost every member of $C$, or
\item every member of $C$ has both a descendant satisfying $P$ and a 
descendant failing $P$.
\end{enumerate}
\end{proposition}

\begin{proof}
Immediate by Proposition~\ref{structureprop}: if neither (1) nor (2) 
holds,
then both the $P$-conformists and the $P$-rebels are infinite in number.
\end{proof}

For example, in an inspecies of birds, a property of type $1$ might be 
``has feathers'', a property of type $2$
might be ``has gills'', and a property of type $3$ might be ``is male'', 
or even ``hatched between Monday and Thursday''.

In our opinion, Proposition~\ref{trichotomy}
theoretically reconciles morphological species notions (recognizable,
ecological, phylogenetic, and so on)
with non-morphological species concepts.
Inspecies are defined
solely based on ancestral relations with reckless disregard for any other 
considerations--
and yet, Morphological Trichotomy (Proposition 7) rigorously
shows that morphological aspects of species are hard to avoid.

If we have understood him correctly, de Queiroz (2007) would say that the 
uniformities
suggested by Proposition~\ref{trichotomy} are what he calls \emph{former 
secondary species
criteria}, which ``can be used to define subcategories of the species 
category-- that is, to recognize different classes of species''
and which, therefore, we submit as evidence verifying we are justified in
referring to infinitary species as a species notion.

\begin{proposition}
\label{nesting}
If two inspecies have infinite intersection, they are equal.  That is, 
distinct inspecies have almost no members in common.
\end{proposition}

\begin{proof}
Let $C_1$ and $C_2$ be two inspecies with $|C_1\cap C_2|=\infty$.  By 
symmetry it's enough to show $C_1\subseteq C_2$.
Let $v$ be an individual in $C_1$.  By Proposition~\ref{structureprop}, 
$v$ has a descendant in $C_1\cap C_2$, hence in $C_2$.
Since $C_2$ is ancestrally closed, $v\in C_2$.
\end{proof}

Sadly, our infinitary genus notion does not have the coveted 
\emph{nesting property} described by Dress et al.
On the other hand, Proposition~\ref{nesting} shows that (ignoring finite 
sets)
inspecies have that property \emph{too} strongly!
By this we mean that if we copy the techniques in Dress et al.~(2010) to 
turn a nested cluster notion into a forest notion,
we get a degenerate forest with all vertices isolated and no arcs at all.
This hints that, where the Tree of Life is concerned, inspecies may be the 
``leaves at infinity,''
while the clusters of Dress et al.~(2010) or the composite species of 
Kornet and McAllister (2005) may be the actual nodes.

As promised in the introduction, we will now prove that the Knight-Darwin 
Law implies a Dual Knight-Darwin Law (see Figure 3).

\clearpage

\begin{multicols}{2}
\begin{theorem}
\label{dualknightdarwin}
Assume the Knight-Darwin Law: that $G$ has no infinite path of vertices 
each with $<2$ \emph{parents}.
Assume also that A1--A4 hold.  Then the following facts are implied:
\begin{enumerate}
\item \emph{(The Dual Knight-Darwin Law)}
No inspecies has an infinite path of vertices each with $<2$ 
\emph{children}.
\item
There are infinitely many individuals each of whom has multiple children.
\item
Within an inspecies, every individual has a pair of descendants that breed 
together.
\end{enumerate}
\end{theorem}

\columnbreak

\vspace{-10pt}
\begin{center}
\includegraphics[scale=.7]{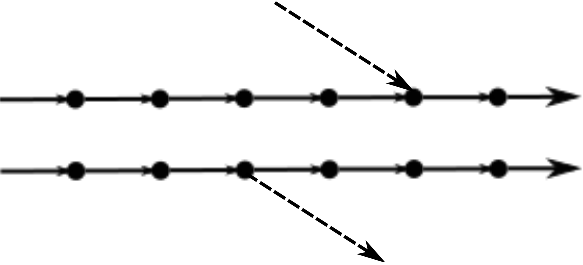}
\end{center}
\vspace{-10pt}
\begin{quote}
Figure 3: The Knight-Darwin Law (top) and the Dual Knight-Darwin Law
(bottom).
\end{quote}
\vspace{-.6in}
\end{multicols}

\begin{proof}
(1)
Let $C$ be an inspecies, and
let $v_1,v_2,\ldots$ be any infinite path in $C$.  We must show some $v_i$ 
has $\geq 2$ children.
For sake of contradiction, assume not.
By the Knight-Darwin Law, there is some $v_{i_1}$ with $\geq 2$ parents.
And then, by the Knight-Darwin Law applied to 
$v_{i_1+1},v_{i_1+2},\ldots$, there is
some $v_{i_2}$ ($i_2>i_1$) with $\geq 2$ parents.
This process continues forever: there are $i_1<i_2<i_3<\cdots$
such that each $v_{i_j}$ has $\geq 2$ parents.
For each $j>1$, let $p_j$ be a parent of $v_{i_j}$ different than 
$v_{i_j-1}$.
All these new parents are in $C$ by ancestral closure, and by Assumption
A2, there are infinitely
many of these new parents.  By Proposition~\ref{structureprop}, one of the 
$p_j$ is descended from $v_1$.
So $v_{i_j}$ has at least two distinct parents, $v_{i_j-1}$ and $p_j$, 
both of whom are descended from $v_1$ (or possibly
one of them could equal $v_1$).
This implies (by acyclicity) that if we look at the path 
$v_1,\ldots,v_{i_j}$, there must have been a fork somewhere,
in order for $p_j$ to be born
and breed with $v_{i_j-1}$.  This contradicts the assumption that none of 
the $v_i$'s have $\geq 2$ children.

(2) Immediate from (1) along with Theorem~\ref{maintheorem}.

(3) Let $v$ be an individual in an inspecies $C$.
By Proposition~\ref{prolific}, there are only finitely many individuals in 
$C$ not descended from $v$,
and by Assumption A2, these finitely many non-$v$-descendants have only
finitely many children.
Thus, there must be some individual in $C$ descended from $v$, both of 
whose parents are also descended from $v$.
Those two parents are therefore a pair of descendants of $v$ which breed 
together, as desired.
\end{proof}

\subsection{A Response to Sturmfels}

B.~Sturmfels asked (2005): ``Can biology lead to new theorems?''
We have self-imposed upon ourselves three \emph{criteria} that we feel obliged to meet in order
to satisfy ourselves in responding.  In our opinion, in order for a mathematical biology paper
to ``Lead to new theorems,'' it should:
\begin{itemize}
\item (T1) Include simple and interesting theorems about some new type of system motivated by biology.
\item (T2) Demonstrate breadth, by using the new system to give original new proofs of some already known results (or 
open problems).
\item (T3) Demonstrate nontriviality, by stating a nontrivial open problem which is not too contrived.
\end{itemize}

If we have erred in these criteria, we have attempted to err on the side of stringency.  T3 could of course be replaced 
by the inclusion of a theorem with an interesting and difficult proof (e.g.~a \emph{non}-open problem).

For T1, we consider Theorem 3 and Propositions 5, 6, 7, and 8 sufficient.  We have given some interesting theorems about 
new systems motivated by biology.

For T2, we consider Theorem 9 partially satisfactory.  We consider it common knowledge that a genealogical human family 
tree must necessarily either keep branching (in the sense of multiple children being born of a parent) or keep bringing 
in new roots via marriage, or both, and if it stops doing so, successive generations will dwindle in size until 
extinction.  Theorem 9 (part 2) formalizes this, and (part 1) (along with Theorem 3) generalizes it.  Likewise, Theorem 9 
(part 3) formalizes, strengthens, and reproves a certain notion that inbreeding is unavoidable.  The original notion is 
that if we assume no inbreeding, then going back (say) 50 generations we should expect a human population of at least 
$2^{50}\approx 10^{15}$ by counting nothing but the distinct great${}^{48}$-grandparents of this author.  To satisfy T2 
further, we have used infinite graphs in systematic biology to give an alternate proof (in (Alexander, preprint)) of a 
result from (Johnston 2010) about Conway's Life-like games.

For T3, we state the following open problem.

\begin{openproblem}
If $s=(s_0,s_1,\ldots)$ is an infinite sequence from the alphabet 
$\{M,F\}$,
say that $s$ is \emph{biologically unavoidable}
if in every possible gendered population conforming to A1--A4 
(\emph{gendered} here
meaning that every non-root has a male parent and a female parent),
there is a sequence of organisms, each a parent of the next, whose
genders match $s$ (the sequence does not need to begin with a root).
What are the biologically unavoidable sequences?
\end{openproblem}

\begin{wrapfigure}{r}[.5in]{0.5\textwidth}
\vspace{-10pt}
\begin{center}
\includegraphics[scale=.5]{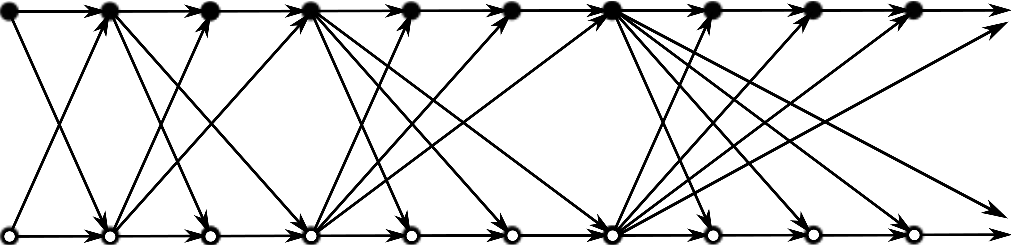}
\end{center}
\vspace{-10pt}
\begin{quote}
Figure 4:
A hypothetical population (modified from an example of
T.~J.~Carlson) in which not every gender sequence is realized.
Solid vertices represent males and open vertices represent females.
One particular sequence absent in this population is
$M^2FM^5F\cdots M^{3n-1}F\cdots$.
\end{quote}
\vspace{-.3in}
\end{wrapfigure}

A priori, it could be that every sequence is biologically unavoidable.
But 
in (Alexander, preprint) we show a counterexample (see Figure 4).
It could also be, a priori, that \emph{no} sequence is biologically 
unavoidable, but in the same paper, we show that every eventually-periodic 
sequence is unavoidable.  It remains an open question even whether there 
is a single
unavoidable sequence not eventually periodic.

Note, the result about every eventually periodic sequence being 
unavoidable could not be true if the sequences of organisms were required
to start at a root, as the following example shows.  According to 
J.~Diamond (1997),
``In an extreme scenario the first settlers [of Australia] are pictured as
$\ldots$ a single pregnant young woman carrying a male fetus''.
Define the set of \emph{Aboriginals} to consist of that Mother, her Child, 
her Child's Father,
and all the joint descendants of the Mother and Child.  Thus, the 
\emph{roots} are exactly the Mother
and the Father.  Even if we assume the Aboriginals satisfy the hypotheses 
of the Problem,
there can be no sequence starting at a root and having genders
$M,F,M,F,\ldots$, for the simple reason that there is only one male root 
and he has no female Aboriginal child.

The theorems in this section are certainly not exhaustive.
We have opted to limit ourselves to theorems
of particular interest in systematic biology.
We hope they are adequate to establish infinitary 
genera and species as at least a potential
candidate for a two-way bridge between biology and the more abstract and 
theoretical side of mathematics.

\section{Infinitary Species and Internodons}
\label{section5}

There are noteworthy similarities between infininitary species and the 
internodal species concepts
of\footnote{This work reflects a stepwise formalization of Hennig's internodal species:
unavoidability of the implied permanency of cleavages (Kornet, 1993), formal implementation
(Kornet et al.~1995); lowering species' status, because of implied short lifespans,
to building blocks (internodons) (Kornet et al.~1995); remedial grouping by secondary
morphological criteria into composite species (Kornet \& McAllister, 2005).  The entire
project was first informally printed as a PhD thesis (\emph{Reconstructing Species; Demarcations in
Genealogical Networks}, 1993, Leiden University).}
Kornet 
(1993), Kornet, Metz, and Schellinx (1995), and 
Kornet and McAllister (2005) (we shall focus on the 1995 paper, as it is 
mathematically the most straightforward).

Informally,
internodons  
are\footnote{To quote Kornet et al.~(1995), internodons are 
``parts of a genealogical network of individual organisms between two 
successive permanent splits or between a permanent split and an extinction 
event.''} the largest clusters subject to the 
constraint that 
\emph{permanent splits}
in the genealogical network give rise to new internodons.  Thus, if (as in 
Figure 1) a branch 
in $G$ splits permanently into two smaller branches, 
near the splitting point there are three internodons: one for the branch 
pre-split, and two for each smaller branch.

The primary similarity
between internodons and inspecies
is the dependence on the \emph{future} 
(to establish a split's permanence may require
infinitely futuristic knowledge).
Both are defined purely from $G$ (and 
birthdates), \emph{sans} morphology.
Both respect permanent splits (individuals on 
opposite sides thereof can share neither an 
internodon nor an inspecies).  And both notions 
arose from attempts to bridge math and biology:
Kornet et al.~attempted to derive new biology from math,
and this author attempted to derive new math from biology.

Our attempt to 
\emph{contrast} inspecies and internodons leads to  
another 
application of infinite graphs to the species problem, the idea of 
benchmark populations.

\subsection{Benchmark Populations}

As the species problem is difficult, one strategy 
might be to
try to solve tiny sub-problems.  A 
\emph{benchmark population} is a hypothetical population, together with a 
question which would be trivial in sight of a species 
problem solution.  One such question might be, ``does this 
population 
include members of multiple species, or just one?''  If we 
solve the species problem, such a question should be 
straightforward.  Til then, these questions are
species sub-problems, which we might hope to answer before answering 
the full species problem.

Interesting benchmark populations will most likely be infinite.
As the following example will show, a good benchmark population is
one which exhibits some ``edge case'' pathology 
specifically meant to push species notions to their limits.

\begin{wrapfigure}{r}[.3in]{0.4\textwidth}
\vspace{-15pt}
\begin{center}
\includegraphics[scale=1]{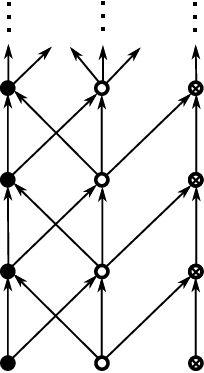}
\end{center}
\vspace{-10pt}
\begin{quote}
Figure 5: The $\frac{1}{3}$ variant population.
\end{quote}
\vspace{-.8in}
\end{wrapfigure}
Our attempt to contrast 
inspecies 
and internodons led us to a particular benchmark population which we 
call the \emph{one-third variant} population\footnote{As candidates one may think of
organisms with complex haploid/diploid life cycles such as social animals (ants, bees, termites, weevils)
or the creative algae, fungi, mosses.  But simpler still would be a genuine sexual network with an additional
\emph{variant male type} with a $y'$ chromosome.  Say, this variant male's sperm cells carrying the $y'$ chromosome 
always 
outcompete its own sperm cells carrying the $x$ chromosome.  Then, per generation we have one (non-variant) \emph{genuine 
$xy$ male} (generating $x$ and $y$ sperm cells), one (non-variant) \emph{genuine $xx$ female} (generating $x$ egg cells 
only), and one \emph{variant $xy'$ male} (generating $y'$ sperm cells that always outcompete its $x$ sperm cells).  In 
this way a \emph{variant} can only co-parent (with the \emph{genuine female}) an $xy'$ variant son.} (see Figure 5).  
This population consists of infinitely many generations, each 
with one \emph{male}, one \emph{female},
and one \emph{variant}.  Each generation's male and 
female produce the next generation's male and female, 
and each generation's female and variant
produce the next generation's 
variant.  The question is:  are the variants
in the same species as the non-variants?

Any well-defined species notion (depending only on graph-theoretical 
considerations) should either answer the above 
question,
explain why the question is ill-posed,
or else use assumptions about reality which rule out the 
$\frac{1}{3}$-variant population.

Internodal species notions say 
that variants are in the same species as non-variants
(there are no splits in sight).  Infinitary species 
say that the variants are \emph{not} the same species: 
in fact the variants are not in any inspecies at 
all
(lest the
directed path of variants would violate the Dual Knight-Darwin Law
(Theorem~\ref{dualknightdarwin} part 1)).

\subsection{Infinitary Species (Type II)}

Kornet et al.~(1995) suggested that internodons
are species building blocks (a suggestion bearing 
fruit in
Kornet \& McAllister (2005)).  This raises one's hopes that perhaps 
infinitary
species are unions of internodons.  The $\frac{1}{3}$-variant population
dashes those hopes.

This motivates a species notion using the same principles
as inspecies but admitting decomposition into internodons.
With the $\frac{1}{3}$-variant population in mind, what must be found is
a sense in which, e.g., a variant is
an ``ancestor'' of later non-variants.

Note that in general $u$ is an ancestor of $v$ if and only if $u$ is older than $v$
and $G$ has a \emph{directed} path from $u$ to $v$ avoiding vertices older than $u$.
So call $u$ an \emph{undirected-ancestor} (or simply an \emph{undirancestor})\footnote{Thus, the undirdescendants of $u$
are precisely the members of the \emph{gross dynasty} of $u$ (minus $u$ and its same-exact-age peers), in the
language of Kornet et al.\ (1995).} of $v$ if $u$ is older than $v$
and $G$ has an \emph{undirected} path from $u$ to $v$ avoiding vertices older than $u$.
Figure 6 shows a small population (left) and the corresponding undirancestral relations (middle).
The latter clause-- that $G$ contains such an undirected path--
is written $u(\mathbf{PC}_{\geq u})v$ in Kornet et al.\ (1995), a notation
we too shall adopt.
This makes variants undirancestors of future non-variants
in the $\frac{1}{3}$-variant population.
Say a set of individuals is \emph{undirancestrally closed}
if it contains all its own undirancestors.

\vspace{-.05in}
\begin{center}
\includegraphics[scale=.7]{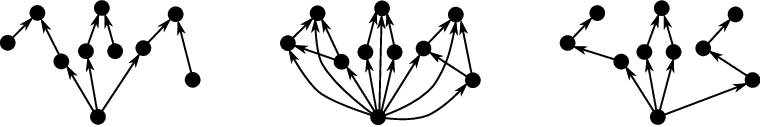}
\end{center}

\vspace{-.2in}
\noindent
\begin{quote}
Figure 6: Left: A small population and its parental relations (birthdates
indicated by vertical height).  Middle: Its undirancestral relations.
Right: Its undirparental relations.
\end{quote}

Say $u$ is an \emph{undirected-parent} (or simply \emph{undirparent}) of $v$ if $u$ is an undirancestor of $v$ and there is
no $w$ such that $u$ is an undirancestor of $w$ and $w$ is an undirancestor
of $v$ (this is a general recipe for reverse-engineering parenthood notions from 
ancestorhood notions).  Figure 6 shows a small population (left) and its undirparent
relations (right).

We'd like to
apply our infinitary species machinery 
to undirparenthood.  But first we must ensure we
have not corrupted any of Assumptions A1--A4.

\begin{lemma}
\label{newGfromold}
Let $G'$ be the graph whose vertices are the organisms of $G$,
and in which an arc is directed from $u$ to $v$ if and only if $u$ is an 
undirparent of $v$; endow $G'$ with the same birthdates as $G$.
Given that $G$ satisfies A1--A4, so does $G'$.
\end{lemma}

Note that in the graph $G'$ of the lemma, for any vertices $u,v\in G$, $u$ 
is an undirparent (resp.~undirancestor) of $v$ in $G$ iff $u$ 
is a parent (resp.~ancestor) of $v$ in $G'$.

\begin{proof}[Proof of Lemma~\ref{newGfromold}]
\emph{$G'$ Satisfies A1.}
First we claim ($*$) that any vertex (say $v$) with a parent (say $u$) must have an undirparent.
If $u$ is an undirparent of $v$, this is trivial.
Otherwise, since $u$ is clearly an undirancestor of $v$, there is some $w$ blocking $u$ from being
an undirparent of $v$.
If $w$ is an undirparent of $v$, we're done; if not, there is some $w'$ blocking $w$ from being
an undirparent of $v$... this process must terminate, because (by A3 in $G$) there are only
finitely many individuals born before $v$.  This proves ($*$).
The contrapositive of ($*$) says: any individual with no undirparent must 
have no parent.
There are only finitely many such individuals, since $G$ satisfies A1.

\emph{$G'$ satisfies A2.}  Let $u\in V$, we must show $u$ has only finitely many undirchildren.
Assume, for sake of contradiction, $u$ has infinitely many undirchildren.
By A3, only finitely many vertices share $u$'s birthdate, and by A2 these have finitely many children,
so, by A3 again, $u$ has an underchild $v$ younger than all children of all vertices with $u$'s birthdate.
Since $u$ is an undirancestor of $v$, there is an undirected path $u=u_0,\ldots,u_n=v$ avoiding
vertices born before $u$.
Let $j<n$ be maximal such that $t(u_j)=t(u)$.
We chose $v$ younger than all children of $u_j$, so $n>j+1$.
Thus it makes sense to pick $j<i<n$ such that $u_i$ is as old as possible among all such choices for $i$.
Since $u_0,\ldots,u_n$ witnesses $u(\mathbf{PC}_{\geq u})v$, it follows that $u_0,\ldots,u_i$ witnesses $u(\mathbf{PC}_{\geq u})u_i$.
By maximality of $j$, $t(u_i)>t(u)$, so this shows $u$ is an undirancestor of $u_i$.
Since $u_i$ was chosen as old as possible among $v_{j+1},\ldots,v_{n-1}$,
$u_i,\ldots,u_{n-1}$ avoids vertices born before $u_i$.
And $u_n=v$ is younger than $u_i$, because $t(u_i)\leq t(u_{j+1})$ (by choice of $i$) and $t(u_{j+1})<t(v)$ (by choice of $v$,
since $u_{j+1}$ is a child of $u_j$, which has $u$'s birthdate).
Thus $u_i,\ldots,u_n$ witnesses $u_i(\mathbf{PC}_{\geq u_i})v$, and since $t(u_i)<t(v)$ this shows $u_i$ is an undirancestor of $v$.
Letting $w=u_i$, this violates the definition of $u$ being an undirparent of $v$, a contradiction as desired.

\emph{$G'$ satisfies A3.} If $u$ is an undirparent of $v$, then in particular $u$ is an
undirancestor of $v$, so by definition $t(u)<t(v)$.  The other part of A3 (the finiteness of $\{v\in V\,:\,t(v)<x\}$)
is trivial since $G'$ has the same birthdates as $G$.


\emph{$G'$ satisfies A4.} $G$ is infinite (by A4), and $G'$ has the same vertices, so $G'$ is infinite.
\end{proof}

One very special case is if we assume no two organisms ever share a birthdate.

\begin{proposition}
Let $G'$ be as in Lemma~\ref{newGfromold}.
If no two vertices share the exact same birthdate, then $G'$ is a forest.
\end{proposition}

\begin{proof}
Assume no vertices share the same birthdate.
To show $G'$ is a forest, it suffices to show $G'$ contains no cycles.  By A3, existence of a cycle in $G'$ implies
some organism $u$ has $\geq 2$ distinct undirparents $v_1,v_2$.
By hypothesis, $t(v_1)\not=t(v_2)$, we may assume $t(v_1)<t(v_2)$.
It follows $v_1$ is an undirancestor of $v_2$ (we have $v_1(\mathbf{PC}_{\geq v_1})v_2$ via $u$);
letting $w=v_2$ violates the definition of $v_1$ being an undirparent of $u$.
\end{proof}

\begin{multicols}{2}
In Figure 7, we induce unique birthdates in the network from Hennig's Figure 4 by slightly tilting it;
thin edges are Hennig's original parental relations,
thick edges are undirparental relations (which make a tree, with precisely one splitting point
exactly where internodal speciation occurs).  Note we assume the minor cleavages late in the network
are temporary (i.e., that the network continues beyond what is shown, and minor branches rejoined later).

By an \emph{infinitary species (type II)} for $G$, I mean an infinitary 
species for the graph $G'$ of Lemma~\ref{newGfromold}.
In other words, an infinitary species (type II) is an infinite set of 
individuals closed under undirancestry and which cannot possibly be 
shrunk while preserving these properties.
With A1--A4, an inspecies (type II) exists by Lemma~\ref{newGfromold}
and Theorem~\ref{maintheorem}.
All the results from Section~\ref{section4} carry over to inspecies (type 
II), with the prefix ``undir'' attached appropriately.
For example, Proposition~\ref{prolific} says that within an inspecies 
(type II), everyone is an undirancestor of almost everyone.
In case no organisms share the exact same birthdate, infinitary species (type II)
are simply the infinite branches in the forest $G'$.

The following theorem relates internodon theory and inspecies.
For the sake of clarity and length, this theorem is stated with 
extreme simplifying assumptions; in future
work we will generalize it to make it more applicable to the real world.

\columnbreak

\vspace{.5in}
\begin{center}
\includegraphics[scale=.7]{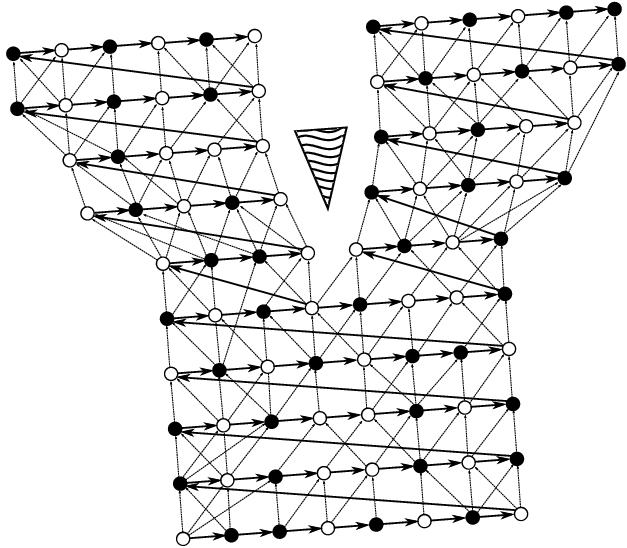}
\end{center}
\vspace{-10pt}
\begin{quote}
Figure 7: Undirparental relations in a copy of Hennig's Figure 4 population, tilted to induce unique birthdates.
We assume the population continues beyond what is shown, in particular that minor cleavages (not indicated by
triangles) are temporary.
\end{quote}
\vspace{-.9in}
\end{multicols}

\begin{theorem}
\label{internodonclimax}
Assume that no two individuals are born simultaneously, and that every individual
has at least one descendant which has more than one parent.
Then every inspecies (type II) is a union of internodons.
\end{theorem}

In Kornet et al.~(1995), the property (above) that an individual has at least one descendant with multiple parents is 
called the $\mathbf{SD}$-property (for Sexual Descendant).

The proof of the theorem is given in Appendix A.  In future work when we 
strengthen this theorem, the strengthened version
will involve the Knight-Darwin Law,
tying together three central themes of this paper.

\section{Discussion}

What began with an attempt to generalize a graph of Dress et al.~by 
considering future organisms as well as present and past, led us to 
consider some existing usages of the future biosphere in systematic 
biology.  We noticed that at least one of these (the Knight-Darwin Law) 
and probably more (Hennig's internodal species concept) hinge upon 
infinitary assumptions, leading us to consider \emph{infinite graphs in 
systematic biology}.  Laboring under the (simplifying) assumption that the 
biosphere will be infinite, we found some natural cluster notions with 
some idealized properties of species.

Defined entirely by ancestral
relations, infinitary species share an \emph{intrinsic} quality with other 
cladistical or non-morphological notions.
And yet, Proposition~\ref{trichotomy} shows despite this, we still enjoy
certain \emph{recognizable} qualities associated with morphological 
species notions.
We consider this an application to the species problem because it 
demonstrates that these seemingly unreconcilable approaches can actually 
meet 
when viewed at sufficient scale.
Methods of Dress et 
al.~suggest (in light of Proposition~\ref{nesting})
that infinitary species play the role of leaves at the distant future 
extremities of the tree of life.

The main flaws of infinitary genera and species are as follows.  First, 
they are completely unobservable:
to establish even one organism's membership in one particular infinitary 
genus, it is necessary to have knowledge
of the infinitely distant future.  Second, not all organisms are 
guaranteed membership in any infinitary species.
A childless organism is
excluded from being in any infinitary species (e.g., by 
Proposition~\ref{structureprop});
this is reminiscent of the cladist's disregard for the sterile mule.
Even if an organism has infinitely many descendants, it is still no 
guarantee of belonging to an infinitary
species (see Section 6.1).
Third, when an organism does have a species, it might not be unique: an 
individual may belong to multiple
infinitary species at once (but the magnitude of this flaw is
limited by Proposition~\ref{nesting}).  These crucial flaws re-emphasize
that while normal species are like nodes on the tree of life, infinitary 
species are more like ``leaves at infinity.''

We compared and contrasted our new infinitary species notion with the 
internodal species concept of Hennig, as formalized by Kornet et al.
To contrast the two, we compared how they performed on a particularly 
extreme population-- a technique which we referred to as using a 
\emph{benchmark population}, and which we hope might serve as a more 
general application of infinitary graph theory to the species problem,
breaking it into more manageable species subproblems.

We've exhibited some new theorems (among 
them, a dualization of
the Knight-Darwin Law) which we find interesting in their own right.
These theorems,
like those of Kornet et al.~(1995) before us,
are of a less computational nature than a lot of mathematical theorems 
which have come from biology.
We hope they may contribute to the two-way bridge between biology and the 
theoretical, philosophical,
combinatorial
side of
mathematics.

\section{Acknowledgments}

We thank Timothy J.~Carlson, Bora Bosna, Mike Fenwick,
and especially the referees
for much productive feedback and discussion.

\appendix
\section{Proof of Theorem~\ref{internodonclimax}}

In this appendix we will prove Theorem~\ref{internodonclimax}
(hereafter we assume its hypotheses).
First we review the definitions in Kornet et 
al.~(1995).
As we state them they are not equivalent; they are altered to
fit the hypotheses of Theorem~\ref{internodonclimax}, allowing us to 
simplify them\footnote{To be precise, we are assuming that every 
individual has the $\mathbf{SD}$ property, which
equates the equivalence relations $\mathbf{INT}$ and $\mathbf{INTSD}$,
as well as the sets $\mathbb{DYN}$ and $\mathbb{GDYN}$,
of the 1995 paper.}.

\begin{definition}
For any $u\in G$, write $\mathbb{DYN}(u)$ for the set
$\{x\in G\,:\,u(\mathbf{PC}_{\geq u})x\}$
and write $\geqq(u)$ for $\{x\in G\,:\,t(x)>t(u)\}$.
If $u,v\in G$ and $t(u)<t(v)$, define
\[
u\mathbf{INT}v :\Leftrightarrow
u(\mathbf{PC}_{\geq u})v
\wedge
\forall r[\{u(\mathbf{PC}_{\geq u})r \wedge \left(\, t(r)\leq 
t(v)\,\right)\}
\Rightarrow (\mathbb{DYN}(u)\cap \geqq(r) = \mathbb{DYN}(r))].
\]
If $t(v)<t(u)$ then $u\mathbf{INT}v:\Leftrightarrow v\mathbf{INT}u$,
and if $t(v)=t(u)$ then $u=v$ by our assumption that no individuals are 
born simultaneously, and we define $u\mathbf{INT}u$.
\end{definition}

The main theorem of Kornet et al.~(1995) implies $\mathbf{INT}$ is an 
equivalence relation.

\begin{definition}
An \emph{internodon} is an $\mathbf{INT}$-equivalence class.
\end{definition}

We could now dive directly into the proof of 
Theorem~\ref{internodonclimax}
but we prefer to factor out a sufficient condition which might be 
useful in the future for proving that other species notions decompose into 
internodons (this condition, too, will be generalized in future work).

\begin{proposition}
\label{sufficientcond}
(Sufficient conditions for internodons-decomposition)
(Assuming the hypotheses of Theorem~\ref{internodonclimax}.)
Let $S\subseteq V$.  If $S$ is undirancestrally closed,
and for each $u\in S$ and $t\in \mathbb{R}$ there is some $u'\in S$
born after $t$ with $u(\mathbf{PC}_{\geq u})u'$; then $S$ is a union of 
internodons.
\end{proposition}

\begin{proof}
It suffices to let $u\in S$, $v\in G$, and show that if $u\mathbf{INT}v$
then $v\in S$.

Case 1: $v$ is born before $u$.
Since $v\mathbf{INT}u$, in particular $v(\mathbf{PC}_{\geq v})u$.
Thus $v$ is an undirancestor of $u$, putting $v\in S$ by
undirancestral closure.

Case 2: $v$ is born after $u$.
Since $u\mathbf{INT}v$, by definition of $\mathbf{INT}$ we have
$u(\mathbf{PC}_{\geq u})v$
and ($*$) for every $r$ such that $u(\mathbf{PC}_{\geq u})r$
and $t(r)\leq t(v)$, we have
$\mathbb{DYN}(u)\cap \geqq(r)=\mathbb{DYN}(r)$.

By the proposition's hypothesis, there is some $u'\in S$, born after $v$,
such that $u(\mathbf{PC}_{\geq u})u'$.
Letting $r=v$ in ($*$), we see (since 
$u(\mathbf{PC}_{\geq u})v$ and $t(v)\leq t(v)$)
that
$\mathbb{DYN}(u)\cap\geqq(v)=\mathbb{DYN}(v)$.
Since $u'\in \mathbb{DYN}(u)\cap\geqq(v)$, this shows
$u'\in\mathbb{DYN}(v)$, that is, $v(\mathbf{PC}_{\geq v})u'$.
Thus $v$ is an undirancestor of $u'$,
so $v\in S$ by undirancestral closure of $S$.
\end{proof}

\begin{proof}[Proof of Theorem~\ref{internodonclimax}]
Let $S$ be an inspecies (type II).
The first hypothesis of Proposition~\ref{sufficientcond} holds since $S$ is undirancestrally closed
by definition.
The second hypothesis of Proposition~\ref{sufficientcond} is immediate by Proposition~\ref{prolific}.
\end{proof}

The reader will have noticed that despite imposing such onerous 
hypotheses on Theorem~\ref{internodonclimax}, we barely seem to have actually used those hypotheses.
Their key use was in simplifying the definition of internodons.


\begin{thebibliography}{99}

\bibitem{alexander}
Alexander S (preprint)  Biologically unavoidable sequences.  Submitted.\\
{\small
\url{http://arxiv.org/abs/1212.0186}
\par
}


\bibitem{baum}
Baum DA (2009) Species as ranked taxa. Syst.~Biol.~58: 74--86.

\bibitem{baumshaw}
Baum DA, Shaw KL (1995) Genealogical perspectives on the species problem.  
In: Hoch PC, Stephenson AC, eds. Experimental and molecular approaches to 
plant biosystematics.  St Louis, MO: Missouri Botanical Garden, 289--303.

\bibitem{cohen}
Cohen J (2004)
Mathematics Is Biology's Next Microscope, Only Better; Biology Is
Mathematics' Next Physics, Only Better.
PLoS Biol.~2(12).  Available:\\
{\small
\url{http://www.plosbiology.org/article/info%3Adoi%2F10.1371%2Fjournal.pbio.0020439}.
\par}

\bibitem{origin}
Darwin C (1872)  On The Origin of Species.  6th Edition.

\bibitem{fdarwin}
Darwin F (1898)  The Knight-Darwin Law.  Nature 58: 630--632.

\bibitem{diamond}
Diamond J (1997)  Guns, Germs, and Steel.  W.~W.~Norton \& Company.

\bibitem{dress}
Dress A, Moulton V, Steel M, Wu T (2010)
Species, clusters and the `Tree of life': A graph-theoretic perspective.
J.~Theor.~Biol.~265: 535--542. Available:\\
{\small 
\url{http://www.math.canterbury.ac.nz/~m.steel/Non_UC/files/research/species_jtb.pdf}.
\par}


\bibitem{gowers}
Gowers T (2008) How to use Zorn's lemma. Available:\\
{\small 
\url{http://gowers.wordpress.com/2008/08/12/how-to-use-zorns-lemma/}.
\par}

\bibitem{hennig}
Hennig W (1966, reprinted 1979). \emph{Phylogenetic Systematics}. Davis D, Zangerl R, translators. Urbana, IL: University
of Illinois Press.

\bibitem{hyde}
Hyde D (2011) Sorites Paradox.
In: Zalta E, editor. \emph{The Stanford Encyclopedia of Philosophy (Winter
2011 Edition)}. Available:\\
{
\url{http://plato.stanford.edu/archives/win2011/entries/sorites-paradox/}.
\par
}

\bibitem{johnston}
Johnston N (2010) The B36/S125 ``$2\times 2$'' Life-Like Cellular 
Automaton. In:
Adamatzky A, editor. Game of Life Cellular Automata. Springer-Verlag.  
Preprint Available:\\
{\small
\url{http://njohns01home.webfactional.com/wp-content/uploads/2010/01/2x2.pdf}.
\par}

\bibitem{konig}
K\H{o}nig D (1936).  Theorie der Endlichen und Unendlichen Graphen.
Leipzig: Akademische Verlagsgesellschaft.


\bibitem{kornet1993}
Kornet D (1993)
Permanent splits as speciation events: a formal reconstruction of the
internodal species concept. J.~Theor.~Biol.~164: 407--435.

\bibitem{kornet1995}
Kornet D, Metz J, Schellinx H (1995) Internodons as equivalence classes
in genealogical networks: building-blocks for a rigorous species concept.
J.~Math.~Bio.~34: 110--122.

\bibitem{kornet2005}
Kornet D, McAllister J (2005) The composite species concept: a rigorous 
basis
for cladistic practice. In: Reydon T, Hemerik L, editors. Current Themes
in Theoretical Biology: A Dutch Perspective. Dordrecht: Springer, 95--127.

\bibitem{nixon}
Nixon K, Wheeler Q (1990)
Another way of looking at the species problem: a reply to de Queiroz and 
Donoghue.  Cladistics 6: 77--81.

\bibitem{dequeiroz}
de Queiroz K (2007)  Species concepts and species delimitation.
Syst. Biol. 56: 879--886.

\bibitem{dequeirozdonoghue}
de Queiroz K, Donoghue J (1988)  Phylogenetic systematics and the species 
problem. Cladistics 4: 317--338.

\bibitem{steel}
Steel M (2007) Tools to construct and study big trees: a mathematical 
perspective.
In: Hodkinson T, Parenell J, Waldren S, editors. Reconstructing the Tree 
of Life:
Taxonomy and Systematics of Species Rich Taxa.  CRC Press (Taylor and 
Francis),
pp.~97--112.  Preprint Available:\\
{\small
\url{http://www.math.canterbury.ac.nz/~m.steel/Non_UC/files/research/tools.pdf}.
\par}

\bibitem{sturmfels}
Sturmfels B (2005)
Can biology lead to new theorems?
Annual report of the Clay Mathematics Institute. Available:
{\small \url{http://math.berkeley.edu/~bernd/ClayBiology.pdf} \par}

\bibitem{velasco}
Velasco J (2008) The internodal species concept: a response to `The tree,
the network, and the species'.  Bio.~J.~of the Linnean Soc.~93: 865--869.


\bibitem{zorn}
Zorn M (1935) A remark on method in transfinite algebra.  
Bull.~Amer.~Math.~Soc.~41: 667--670.

\end{thebibliography}
\end{document}